\documentclass[12pt]{amsart}

\usepackage{amsmath,amssymb,amsfonts,amsthm,mathrsfs,epsfig}
\usepackage{graphicx}
\usepackage[usenames,dvipsnames]{color}
\usepackage{hyperref}

\usepackage[margin=1in]{geometry}

\begin{document}
\newtheorem{theorem}{Theorem}[section]
\newtheorem{lemma}[theorem]{Lemma}
\newtheorem{definition}[theorem]{Definition}
\newtheorem{conjecture}[theorem]{Conjecture}
\newtheorem{proposition}[theorem]{Proposition}
\newtheorem{algorithm}[theorem]{Algorithm}
\newtheorem{corollary}[theorem]{Corollary}
\newtheorem{observation}[theorem]{Observation}
\newtheorem{claim}[theorem]{Claim}
\newtheorem{problem}[theorem]{Open Problem}
\newtheorem{remark}[theorem]{Remark}
\newcommand{\noin}{\noindent}
\newcommand{\ind}{\indent}
\newcommand{\om}{\omega}
\newcommand{\I}{\mathcal I}
\newcommand{\N}{{\mathbb N}}
\newcommand{\Z}{{\mathbb Z}}
\newcommand{\LL}{\mathbb{L}}
\newcommand{\R}{{\mathbb R}}
\newcommand{\E}[1]{\mathbb{E}\left[#1 \right]}
\newcommand{\V}{\mathbb Var}
\newcommand{\Prob}{\mathbb{P}}
\newcommand{\eps}{\varepsilon}
\newcommand{\bsigma}{{\boldsymbol\sigma}}
\newcommand{\bm}{{\boldsymbol m}}
\newcommand{\bu}{{\boldsymbol u}}
\newcommand{\bv}{{\boldsymbol v}}
\newcommand{\tU}{{\mathtt U}}
\newcommand{\tD}{{\mathtt D}}
\newcommand{\tL}{{\mathtt L}}
\newcommand{\tR}{{\mathtt R}}

\newcommand{\Tv}{P}

\newcommand{\mT}{\mathcal{T}}
\newcommand{\mS}{\mathcal{S}}
\newcommand{\mA}{\mathcal{A}}
\newcommand{\mB}{\mathcal{B}}

\newcommand{\Cyc}[1]{\mathrm{Cyc}\left(#1\right)}
\newcommand{\Seq}[1]{\mathrm{Seq}\left(#1\right)}
\newcommand{\Mul}[1]{\mathrm{Mul}\left(#1\right)}
\newcommand{\Mulo}[1]{\mathrm{Mul}_{>0}\left(#1\right)}
\newcommand{\Set}[1]{\mathrm{Set}\left(#1\right)}
\newcommand{\Setd}[1]{\mathrm{Set}_{d}\left(#1\right)}
\newcommand{\Bin}{\mathrm{Bin}}

\title{A probabilistic version of the game of Zombies and Survivors on graphs}
\author{Anthony Bonato}
\address{Department of Mathematics, Ryerson University, Toronto, ON, Canada}
\email{abonato@ryerson.ca}
\author{Dieter Mitsche}
\address{Universit\'{e} de Nice Sophia-Antipolis, Nice, France}
\email{dmitsche@unice.fr}
\author{Xavier P\'{e}rez-Gim\'{e}nez}
\address{Department of Mathematics, Ryerson University, Toronto, ON, Canada}
\email{xperez@ryerson.ca}
\author{Pawe\l{} Pra\l{}at}
\address{Department of Mathematics, Ryerson University, Toronto, ON, Canada}
\email{\tt pralat@ryerson.ca}

\thanks{The authors gratefully acknowledge support from NSERC}

\begin{abstract}
We consider a new probabilistic graph searching game played on graphs, inspired by the familiar game of Cops and Robbers. In Zombies and Survivors, a set of zombies attempts to eat a lone survivor loose on a given graph. The zombies randomly choose their initial location, and during the course of the game, move directly toward the survivor. At each round, they move to the neighbouring vertex that minimizes the distance to the survivor; if there is more than one such vertex, then they choose one uniformly at random. The survivor attempts to escape from the zombies by moving to a neighbouring vertex or staying on his current vertex. The zombies win if eventually one of them eats the survivor by landing on their vertex; otherwise, the survivor wins. The zombie number of a graph is the minimum number of zombies needed to play such that the probability that they win is strictly greater than 1/2. We present asymptotic results for the zombie numbers of several graph families, such as cycles, hypercubes, incidence graphs of projective planes, and Cartesian and toroidal grids.
\end{abstract}

\maketitle

\section{Introduction}

A number of variants of the popular graph searching game Cops and Robbers have been studied. For example, we may allow a cop to capture the robber from a distance $k$, where $k$ is a non-negative
integer~\cite{bonato5,bonato4}, play on edges~\cite{pawel}, allow the robber to capture a cop~\cite{bonato0}, allow one or both players to move with different speeds~\cite{NogaAbbas,fkl} or to
teleport, allow the robber to capture the cops~\cite{bonato0}, have the cops move one at a time~\cite{bbkp,bbkp2,oo}, have the cops play on edges and the robber on vertices~\cite{km, contain}, or make the robber invisible or drunk~\cite{drunk1,drunk2}. For additional background on Cops and Robbers and its variants, see the
book~\cite{bonato} and the surveys~\cite{bonato1,bonato2,bonato3}.

\bigskip

For a given connected graph $G$ and given $k \in \N$, we consider the following probabilistic variant of Cops and Robbers, which is played over a series of discrete time-steps. In the game of \emph{Zombies and
Survivors}, suppose that $k$ \emph{zombies} (akin to the cops) start the game on random vertices of $G$; each zombie, independently, selects a vertex uniformly at
random to start with. Then the \emph{survivor} (akin to the robber) occupies some vertex of $G$. As zombies have limited intelligence, in each round, a given zombie moves towards the survivor along a
shortest path connecting them. In particular, the zombie decreases the distance from its vertex to the survivor's. If there is more than one neighbour of a given zombie that is closer to the survivor
than the zombie is, then they move to one of these chosen uniformly at random. Each zombie moves independently of all other zombies. As in Cops and Robbers, the survivor may move to another
neighbouring vertex, or \emph{pass} and not move. The zombies win if one or more of them \emph{eat} the survivor; that is, land on the vertex which the survivor currently occupies. The survivor, as survivors should do in the event of a zombie attack, attempts to survive by applying an optimal strategy; that is, a strategy that minimizes the probability of being captured. Note that there is no strategy for the zombies; they merely move on
geodesics towards the survivor in each round. Note that since zombies always move toward the survivor, he can pass at most $D$ times, where $D$ is a diameter of $G$, before being eaten by some zombie.
Note also that the game can be extended to the case of $G$ being disconnected, by having zombies that lie in connected components of $G$ different from that of the survivor simply follow a random walk. Nevertheless, in this paper we will only consider connected graphs. We note also that our probabilistic version of Zombies and Survivors was inspired by a deterministic version of this game (with similar rules, but the zombies may choose their initial positions, and also choose which shortest path to the survivor they will move on) first considered in~\cite{hm}.

\bigskip

Let $s_k(G)$ be the probability that the survivor wins the game, provided that he follows the optimal strategy. Clearly, $s_k(G)=1$ for $k < c(G)$, where $c(G)$ is the cop number of $G$.
On the other hand, $s_k(G) < 1$ for any $k \ge c(G)$, since with positive probability the zombies may follow an optimal cop strategy. Usually, $s_k(G) > 0$ for any $k \ge c(G)$; however, there are some examples of graphs for which $s_k(G) = 0$ for every $k \ge c(G)$ (consider, for examples, trees).
Further, note that $s_k(G)$ is a non-decreasing function of $k$ (that is, for every $k \ge 1$, $s_{k+1}(G) \le s_k(G)$), and $s_k(G) \to 0$ as $k\to \infty$. The latter limit follows since the probability that each vertex is initially occupied by at least one zombie tends to 1 as $k \to \infty$.

Define the \emph{zombie number} of a graph $G$ by
$$
z(G) = \min \{ k \ge c(G) : s_k(G) \le 1/2 \} .
$$
This parameter is well defined since $\lim_{k\to\infty}s_k(G)=0$. In other words, $z(G)$ is the minimum number of zombies such that the probability that they eat the survivor is strictly greater than 1/2. The ratio $Z(G) = z(G)/c(G) \ge 1$ is the \emph{cost of being undead}.  Note that there are examples of families of graphs for which there is no cost of being undead; that is, $Z(G)=1$ (as is the case if $G$ is a tree).

On the other hand, $Z(G)$ can be of order as large as the number of vertices, as the following example shows. Let $G$ be a graph consisting of a $5$-cycle with vertices $v_i,$ where $1\le i \le 5$,
and $n-5$ leaves attached to $v_1$, as shown in Figure~\ref{fig:large}. Although two cops suffice to capture the robber in this graph, many more zombies are needed to eat the survivor. To see this, suppose
that the game is played against $k$ zombies. With probability
$$
\left( 1- \frac{5}{n} \right)^{k} = \exp \left( - \frac {5k}{n} + O \left( \frac {k}{n^2} \right) \right) > \frac 12
$$
all zombies start outside the cycle, provided that, say, $k = \frac {\log 2}{5} \ n - \sqrt{n}$ and $n$ is large enough. The survivor chooses to start at $v_2$, and then all zombies immediately move to $v_1$. The survivor
continues to walk around the cycle, and all zombies chase him without eating him forever. On the other hand, if $k = \frac {\log 2}{5} \ n + \sqrt{n}$, then with probability more than $1/2$, at least one
zombie starts on the cycle and at least one starts on the leaves. Conditioning on this event, the survivor loses with probability 1. We obtain that $z(G) \sim \frac {\log 2}{5} \ n$ (where $f(n) \sim g(n)$ denotes that $\lim_{n \to \infty} f(n)/g(n) = 1$), and so $Z(G) \sim
\frac {\log 2}{10} \ n$.
\begin{figure} [h!]
\begin{center}
\epsfig{figure=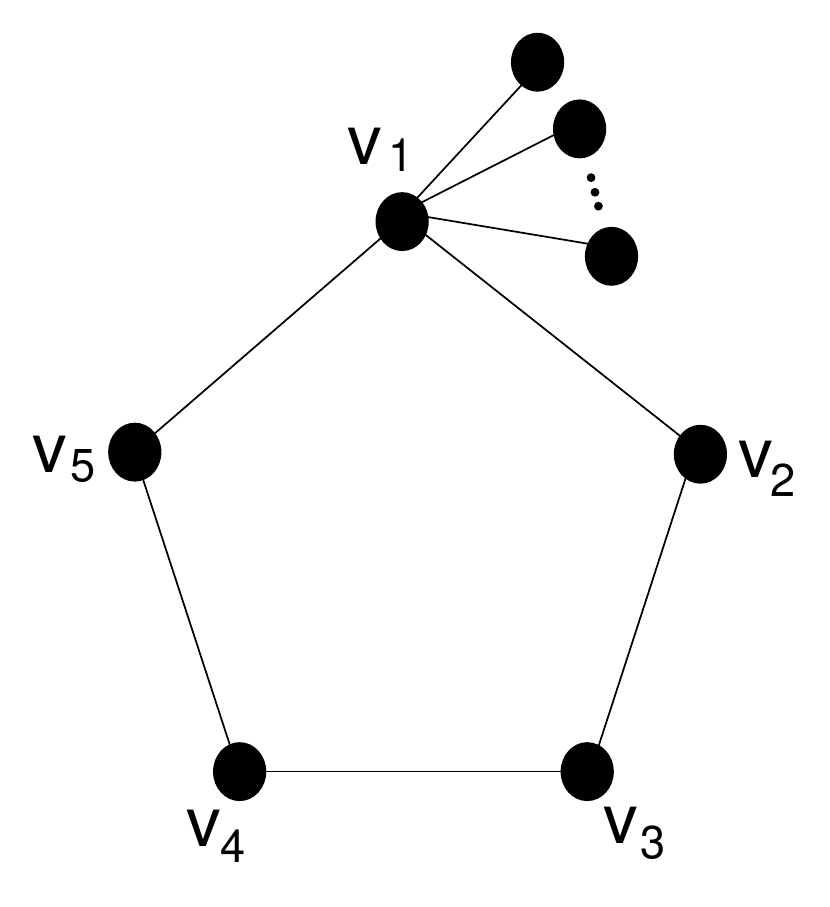,scale=0.6}
\caption{Example of a graph with a large cost of being undead. }\label{fig:large}
\end{center}
\end{figure}

\bigskip

The paper is organized as follows. In Section~\ref{seccycle}, we begin our discussion with cycle graphs.  Theorem~\ref{cycle} gives the asymptotic value of the zombie number of cycles. In
Section~\ref{secproj}, we consider the zombie number of the incidence graphs of projective planes. By using double exposure and coupon collector problems, we show in Theorem~\ref{proj} that about two
times more zombies are needed to eat the survivor than cops. We consider hypercubes $Q_n$ in Section~\ref{sechyper}, and show in Theorem~\ref{hyper} that $z(Q_n) \sim \frac23
n$, as $n \to \infty$. The final section considers both Cartesian grids and grids formed by products of cycles (so called \emph{toroidal grids}). In toroidal grids, we prove in
Theorem~\ref{thm:torus} a lower bound for the zombie number of $\sqrt n/(\omega\log n)$, where $\omega = \omega(n)$ is going to infinity as $n \to \infty$. The proof relies on the careful analysis of a
strategy for the survivor.

\bigskip

Throughout, we will use the following version of \emph{Chernoff's bound}. For more details, see, for example,~\cite{JLR}. Suppose that $X \in \Bin(n,p)$ is a binomial random variable with expectation
$\mu=np$. If $0<\delta<1$, then
$$
\Prob [X < (1-\delta)\mu] \le \exp \left( -\frac{\delta^2 \mu}{2} \right),
$$
and if $\delta > 0$,
\[\Prob [ X > (1+\delta)\mu] \le \exp\left(-\frac{\delta^2 \mu}{2+\delta}\right).\]

The above bounds show that with high probability $X$ cannot be too far away from its expectation. However, it is also true that with high probability $X$ cannot be too close to $\E{X}$. We will use this fact only for the $p=1/2$ case. Let $X \in \Bin(n,1/2)$. First, let us use Stirling's formula ($k! \sim \sqrt{2\pi k} (k/e)^k$) and observe that for each $t$ such that $0 \le t \le n$ we have
$$
\Prob [ X = t] \le \Prob \Big[ X = \left \lfloor n/2 \right \rfloor \Big] = \frac {{ n \choose \lfloor n/2 \rfloor }} {2^n} \sim \sqrt{\frac {2}{\pi n} } < \frac {1}{\sqrt{n}}.
$$
Hence, for each $\eps>0$ there exists $c = c(\eps) > 0$ such that
\begin{equation}\label{eq:Bin_lower}
\Prob \Big[ |X - n/2| < c \sqrt{n} \Big] < \eps.
\end{equation}

\bigskip

For a reference on graph theory the reader is directed to \cite{west}. For graphs $G$ and $H$, define the \emph{Cartesian product} of $G$ and $H$, written $G\square H,$ to have vertices $V(G)\times
V(H),$ and vertices $(a,b)$ and $(c,d)$ are joined if $a=c$ and $bd \in E(H)$ or $ac \in E(G)$ and $b=d.$ Many results in the paper are asymptotic in nature as $n \rightarrow \infty$. We emphasize that the notations $o(\cdot)$ and $O(\cdot)$ refer to functions of $n$, not necessarily positive, whose growth is bounded. We say that an event in a probability space holds
\emph{asymptotically almost surely} (or \emph{a.a.s.}) if the probability that it holds tends to $1$ as $n$ goes to infinity. Finally, as already mentioned, for simplicity we will write $f(n) \sim g(n)$ if $f(n)/g(n) \to
1$ as $n \to \infty$; that is, when $f(n) = (1+o(1)) g(n)$.

\section{Cycles}\label{seccycle}
We analyze the case of cycles $C_n$ first, serving as a warm-up for more complex graph classes studied later in the paper.
Even for elementary graphs such as cycles, unusual situations may arise which never occur in Cops and Robbers. For instance, suppose that all the zombies are initially located on an induced subpath
containing at most $\lceil n/2 \rceil -2$ vertices.  Then the survivor can win by starting at a vertex at distance 2 from the subpath, and move away from the zombies so that all zombies walk towards the same direction.
The large horde of zombies eternally lags behind the survivor who remains safe indefinitely. Otherwise, the zombies win with probability $1$. 

The least obvious case to consider occurs when $n$ is even and all the zombies are initially located on an induced subpath containing $n/2-1$ vertices. The survivor should start at distance 2 from the subpath (otherwise, the two extreme zombies behave like traditional cops and move in opposite directions). With probability at most $1/2$,  zombies make a ``bad move''; that is, all zombies move in the same direction. Note that it is at most 1/2 but not exactly 1/2, as more than one zombie might occupy the two extreme vertices.  But this situation forces the survivor to make a move away from them. As a result, we arrive in an analogous configuration of zombies, where the zombies make another bad move with probability at most $1/2$. With probability 1 in some future round, the extreme zombie will not make a bad move and the survivor will eventually be eaten.

\bigskip

In view of this, we have the following lemma which gives the probability that the survivor wins against $k\ge2$ zombies. (For $k=1$, trivially $s_1(C_3)=0$, and  $s_1(C_n)=1$ for $n\ge4$.)

\begin{lemma}\label{lem:cycle}
For any natural numbers $k \ge2$ and $n\ge 9$, we have that
$$
k \left( \frac 12 - \frac {4}{n} \right)^{k-1} \le s_k(C_n) < k \left( \frac 12 \right)^{k-1}.
$$
In particular, $s_k(C_n) \sim k(1/2)^{k-1}$, as $n \to \infty$.
\end{lemma}

\begin{proof}
Recall that the survivor has a strategy to live forever if and only if all zombies initially lie in a subpath of $r \le \lceil n/2
\rceil -2$ vertices. In order to bound $s_k(C_n)$ from above, we first over-count these configurations, by distinguishing two zombies one at each end of the subpath.
There are $k(k-1)$ ways to select the two distinguished zombies that are placed at the ends of a subpath consisting of $r \le \lceil n/2
\rceil -2$ vertices, and $n$ choices for the position of the path. The two zombies start at the right place with probability $(1/n)^2$, the remaining ones must start on the subpath, which happens
with probability $(r/n)^{k-2}$. Since we are overcounting configurations, it follows that
\begin{eqnarray*}
s_k(C_n) &\le &\sum_{r=1}^{\lceil n/2 \rceil -2} k (k-1) n (1/n)^2 (r/n)^{k-2} \\
&< &k(k-1) n^{1-k} \int_1^{n/2} x^{k-2} dx \\&< &k \left( \frac 12 \right)^{k-1}.
\end{eqnarray*}
On the other hand, by considering configurations in which the two end-vertices of the subpath are occupied by exactly one zombie we have that
\begin{eqnarray*}
s_k(C_n) &\ge& \sum_{r=3}^{\lceil n/2 \rceil -2} k (k-1) n (1/n)^2 ((r-2)/n)^{k-2}\\
& \ge& k(k-1) n^{1-k} \left( \int_0^{n/2-4} x^{k-2} dx \right) \\
&= &k \left( \frac 12 - \frac {4}{n} \right)^{k-1}.
\end{eqnarray*}
The proof of the lemma follows. Moreover, note that,
for $k=2$ and $k=3$, it is easy to find the winning probability precisely:
\begin{align*}
s_2(C_n) &= \sum_{r=2}^{\lceil n/2 \rceil -2} 2\cdot n \cdot \frac {1}{n^2} + n \cdot \frac {1}{n^2}, \mbox{ and } \\
s_3(C_n) &= \sum_{r=3}^{\lceil n/2 \rceil -2} 3\cdot 2\cdot n \cdot \frac {1}{n^2} \cdot \frac {r-2}{n} + 2 \sum_{r=2}^{\lceil n/2 \rceil -2} 3\cdot n \cdot \frac {1}{n^3}  + n \cdot \frac {1}{n^3}. 
\qedhere
\end{align*}
\end{proof}

As an immediate consequence of Lemma~\ref{lem:cycle}, we deduce that $z(C_n) = 4$ for all $n \ge 44$. Combining this with a direct examination of the smaller values of $n$, we derive the zombie number and the cost of being undead for the cycle $C_n$ of length $n\ge3$:
\begin{theorem}\label{cycle}
$$
z(C_n) =
\begin{cases}
4 & \text{ if } n\ge 27 \text{ or } n=23, 25,\\
3 & \text{ if } 11 \le n \le 22 \text{ or } n=9, 24, 26, \\
2 & \text{ if } 4 \le n \le 8 \text{ or } n = 10,\\
1 & \text{ if } n = 3;
\end{cases}
$$
and therefore,
$$
Z(C_n) =
\begin{cases}
2 & \text{ if } n\ge 27 \text{ or } n=23, 25,\\
3/2 & \text{ if } 11 \le n \le 22 \text{ or } n=9, 24, 26, \\
1 & \text{ if } 3 \le n \le 8 \text{ or } n = 10.
\end{cases}
$$
\end{theorem}

\section{Projective Planes}\label{secproj}

Incidence graphs are useful in constructing graphs with large cop numbers; see~\cite{Designs,p}. An \emph{incidence structure} consists of a set $P$ of points, and a set $L$ of lines along with an
incidence relation consisting of ordered pairs of points and lines. Given an incidence structure $S$, we define its \emph{incidence graph} $G(S)$ to be the bipartite graph whose partite sets are the points and lines, respectively, with a point joined to a line if the two are incident in $S.$  Projective planes are some of the most well-studied examples of incidence
structures. A \emph{projective plane} consists of a set of points and lines satisfying the following axioms.

\begin{enumerate}
\item There is exactly one line incident with every pair of distinct points.

\item There is exactly one point incident with every pair of distinct lines.

\item There are four points such that no line is incident with more than two of them.
\end{enumerate}
Finite projective planes have $q^{2}+q+1$ points and $q^{2}+q+1$ lines, for some integer $q>0$ (called the \emph{order} of the plane). Note that each point is incident with $q+1$ lines, and each line is incident with $q+1$ points.
It is known that, for every $q$ that is a prime power, a projective plane $P_q$ of order $q$ exists. The existence of finite projective planes of other orders is an open question. For more on projective planes, see for example \cite{dem}.

\bigskip

The \emph{girth} of a graph $G$ is defined as the length of a shortest cycle. As proved in~\cite{af}, if the girth of $G$ is at least $5$, then $c(G) \ge \delta(G)$, where $\delta(G)$ is the
\emph{minimum degree} of $G$. Let $G_q =G(P_q)$ be the incidence graph of a projective plane of order $q$. It immediately follows from the above properties of $P_q$ that $G_q$ is connected, has girth $6$, is $(q+1)$-regular, and has $2(q^2 + q + 1)$ many vertices. Hence, $c(G_q) \ge q+1$. In fact, $c(G_q)=q+1$, as was shown in~\cite{p}.

\medskip

We will now show that roughly two times more zombies are needed to eat the survivor than cops capturing the robber.

\begin{theorem}\label{proj}
$z(G_q) = 2q + \Theta(\sqrt{q})$, as $q \to \infty$. Hence, $Z(G_{q}) \sim 2$, as $q \to \infty$.
\end{theorem}

\noindent
We will first prove the weaker statement; namely:
$$
2q - \omega \sqrt{q} \le z(G_q) \le 2q + \omega \sqrt{q},
$$
where $\omega = \omega(q)$ is any function tending to infinity as $q\to \infty$,
by considering the upper and lower bounds separately.
After those proofs, we will discuss how to improve the error terms in the estimate of $z(G_q)$. We prove the lower and upper bounds independently, proving the lower bound first.

\begin{lemma}\label{lem1}
Let $k = 2q + \omega \sqrt{q}$, where $\omega = \omega(q)$ is any function tending to infinity as $q\to \infty$. Then $s_k(G_q) \to 0$, as $q \to \infty$.
\end{lemma}
\begin{proof}
As usual, the initial position of the zombies affects the rest of the game. Suppose that $k_\ell$ zombies start on lines and $k_p$ ones on points. First, let us show that if $\min\{k_\ell, k_p\} \ge q$,
then zombies win the game with probability 1. This claim holds for any $q$, not necessarily large. Without loss of generality, the survivor starts the game on the line $v$. We assume that line $v$ is free of
zombies and is not incident with any point containing a zombie, since otherwise the zombies trivially win (recall that zombies move first). Zombies are partitioned into two groups, the first group contains zombies initially occupying lines, the second group consists of zombies initially on points.

For each zombie from the first group, there exists a unique point that is incident with line $v$ and with the line occupied by the zombie. Hence, zombies from the first group move to the corresponding
neighbour of $v$, forcing the survivor to move in the next round. If the survivor survives the first round, all zombies from the first group meet at $v$ and will keep chasing the survivor, always moving to a vertex previously
occupied by the survivor.

Let us now investigate the behaviour of the zombies from the second group. Consider any given zombie of this group. For each point $u$ that is a neighbour of line $v$, there is a unique line that contains
both $u$ and the point occupied by the zombie.
This defines a path of length $3$ joining the positions of the survivor and the zombie via $u$.
Since the girth of $G_q$ is 6, there are precisely $q+1$ edge-disjoint such paths, one for each neighbour of $v$.
The zombie has to select to move along one of these paths uniformly at random and, by doing so, blocks the neighbour $u$ of $v$ corresponding to the chosen path as a potential next move for the survivor (that is, if the survivor moves the $u$, he will get eaten in the next round). This situation will occur at each round until the survivor is eaten.

Recall the survivor has to keep ``running forward'', since he always has zombies from the first team right behind him. So in each round, the survivor has $q$ neighbours to choose from. But
each time, regardless of the history of the process, with positive probability the zombies from the second team block all of these neighbours. Hence, with probability 1 it must happen that, sooner or later, all of
them are blocked.

The rest of the proof is straightforward. Clearly, $\E {k_\ell} = k/2 = q + \omega \sqrt{q} / 2$. It follows from Chernoff's bound that a.a.s.\
$$
| k_{\ell} - \E {k_\ell} | \le \sqrt{\omega \E {k_\ell}} \sim \sqrt{\omega q},
$$
and so a.a.s.\ both $k_{\ell}$ and $k_p$ are at least $q$.
\end{proof}

We next turn to the upper bound of the zombie number of $G_q$.

\begin{lemma}\label{lem2}
Let $k = 2q - \omega \sqrt{q}$, where $\omega = \omega(q)$ is any function tending to infinity as $q\to \infty$. Then $s_k(G_q) \to 1$, as $q \to \infty$.
\end{lemma}
\begin{proof}
We keep the notation introduced in the previous lemma. We will use a well known technique of \emph{double exposure}. For each zombie, at this point we only decide whether the zombie starts on some
line or on some point, each of which happens with probability 1/2. Arguing as before, we obtain that a.a.s.\ both $k_{\ell}$ and $k_p$ are at most $q-1$, and condition on that for the rest of the argument. (Let us note that we will only use the fact
that $k_\ell$ is at most $q-1$; the argument is still valid for $k_p$ larger than $q-1$ as long as it is $(1+o(1))q$; see the comments after the proof.) Now, we can expose the initial positions of the zombies
from the lines, the first group. Regardless where they start, there is at least one point that does not belong to any line associated with zombies (since the number of points, $q^2+q+1$, is more than
$(q-1)(q+1)$, which is a trivial upper bound for the number of points that belong to lines associated with zombies). The survivor starts on one of them, point $v$.
Arguing as in the proof of the previous lemma, during the first round, each zombie from the first team moves to a random neighbour blocking precisely one neighbour of $v$ (uniformly at random).

Let us then expose the initial positions of the zombies from the second team and investigate their behaviour during the first round. Each point other than $v$ has a unique common neighbour with $v$.
Hence, the set of points can be partitioned into $v$ and $q+1$ sets of size $q$ that correspond to $q+1$ neighbours of $v$. With probability $(1-1/(q^2+q+1))^{q-1} \sim 1$ no zombie starts at $v$, and
so we condition on this, pretending that each zombie starts on a random point other than $v$, selected uniformly at random. If a zombie starts on a point corresponding to a neighbour $u$ of $v$,
then this zombie moves to $u$ during the first round, blocking this neighbour as a potential move for the survivor. Hence, each zombie, independently, regardless whether he starts on a line or on a
point, blocks one neighbour of $v$ uniformly at random. Therefore, we obtain a classic \emph{coupon collector problem} with $q+1$ coupons and at most $2(q-1)$ draws. It is straightforward to see that a.a.s.\ at least $(1+o(1)) e^{-2} q$ neighbours of $v$ will not be blocked. Indeed, the probability that a given neighbour of $v$ is not blocked is at least $(1-1/(q+1))^{2(q-1)} \sim e^{-2}$. As a result, a.a.s.\ the survivor can survive the first round.

The rest of the proof is obvious. All zombies from the second team group together and chase the survivor forcing the survivor to ``move forward'' but no other neighbour of the survivor is blocked.
Since there are at most $q-1$ neighbours blocked by the first team (deterministically), the survivor keeps running forever, winning the game.
\end{proof}

The proof of Theorem~\ref{proj} now follows from Lemmas~\ref{lem1} and \ref{lem2}. If $k_+ = k_+(q) = 2q + C \sqrt{q}$ for some large constant $C$, then it follows from Chernoff's bound that with probability at least, say, 2/3, both $k_{\ell}$ and $k_p$ are at least $q$. Hence, we derive that $s_{k_+}(G_q) \le 1/3$ as $q \to \infty$.  On the other hand, if $k_- = k_-(q) = 2q + c \sqrt{q}$ for some sufficiently small constant $c$, then~(\ref{eq:Bin_lower}) implies that with probability at least, say, 2/3, one of $k_{\ell}$ and $k_p$ is at most $q-1$ (and the other one is $(1+o(1))q$).  This time we obtain that
$s_{k_-}(G_q) \ge 2/3+o(1)$ as $q \to \infty$. Note that this time we need to add the $o(1)$ term which corresponds to the probability that the survivor cannot survive the first phase of the game; that is, before all zombies from the second team group together.

\section{Hypercubes}\label{sechyper}

We now investigate the \emph{hypercube of dimension} $n$, written $Q_n$. Note that each vertex of $Q_n$ can be identified with a binary $n$-dimensional vector (or \emph{bit string}). It was established in~\cite{mm} that the cop
number of the Cartesian product of $n$ trees is $\lceil \frac{n+1}{2} \rceil$; in particular, $c(Q_n) = \lceil \frac{n+1}{2} \rceil$. We note that cop numbers of the Cartesian and other graph products were investigated first in~\cite{nn}. 

We will show that approximately $4/3$ times more zombies are needed to eat the survivor.

\begin{theorem}\label{hyper}
$z(Q_n) = \frac {2n}{3} + \Theta(\sqrt{n})$, as $n \to \infty$. Hence, $Z(Q_n) \sim \frac 43$, as $n \to \infty$.
\end{theorem}

\noindent
As we did for the incidence graphs of projective planes, we will first prove the following weaker statement
$$
\frac {2n}{3} - \omega\sqrt{n} \le z(Q_n) \le \frac {2n}{3} + \omega\sqrt{n},
$$
where $\omega = \omega(n)$ is any function tending to infinity as $n\to \infty$. We will then discuss how to improve the error term.  We prove the lower and upper bounds independently, with the lower bound addressed first.

\begin{lemma}\label{lemm1}
Let $k = \frac 23 n -  \omega \sqrt{n}$, where $\omega = \omega(n)$ is any function tending to infinity as $n\to \infty$. Then $s_k(Q_n) \to 1$, as $n \to \infty$.
\end{lemma}
\begin{proof}
Our goal is to show that a.a.s.\ the survivor can avoid being captured when playing against $k$ zombies. First, observe that by Chernoff's bound, a.a.s.\
$$
k/2 + O(\sqrt{\omega n}) = n/3 - (1+o(1)) \omega \sqrt{n} / 2 < n/3
$$
zombies start on vertices having an even number of ones in their binary representations (and, as a result, also less than $n/3$ zombies start with an odd number of ones). Since all zombies
continuously move, this property will hold throughout the game. Hence, the survivor, independently of whether he moves or not, has always an even distance to less than $n/3$ zombies, and also an odd
distance to less than $n/3$ zombies.

The survivor's strategy is the following: he picks as a starting vertex an arbitrary vertex at distance at least $2$ from all the zombies. (This can be easily done as there are, trivially, at most
$k (n+1) < 2^n$ vertices at distance at most 1 from some zombie.) If, immediately after the zombies' move, no zombie is at distance $1$, then the survivor stands still. On the other hand, if the
survivor has a zombie in their neighbourhood, then he wants to move to a safe vertex that is not occupied nor adjacent to any zombie.

Note that there are less than $n/3$ zombies at distance $1$ from the survivor and also less than $n/3$ zombies at distance $2$. Moreover, each zombie at distance $1$ forbids one coordinate, and each zombie at distance $2$ forbids two coordinates, so less than $n$ coordinates are forbidden in total. Hence, the survivor has at
least one coordinate to escape to, and survives for at least one more round. The survivor continues applying the same strategy, and the proof is finished.
\end{proof}

Next, we consider the upper bound.

\begin{lemma}\label{lemm2}
Let $k = \frac 23 n + \omega \sqrt{n}$, where $\omega = \omega(n)$ is any function tending to infinity as $n\to \infty$. Then $s_k(Q_n) \to 0$, as $n \to \infty$.
\end{lemma}
\begin{proof}
This time, our goal is to show that a.a.s.\ $k$ zombies can win. As before, it follows from Chernoff's bound that a.a.s.\ at least $n/3$ zombies start in both positions having an even number of ones
and an odd number of ones, and this property remains true throughout the game. Denote by $d(j,t)$ the graph distance between the $j$th zombie and the survivor after the $t$-th round, and let
$\vec{d}(t)$ be the corresponding $k$-dimensional vector of distances at time $t$. Since $\vec{d}(t)$ is coordinate-wise non-increasing, it suffices to show that given any starting position (for both the survivor
and the zombies) there is a positive probability that after a finite number of steps the distance vector decreases in at least one coordinate. Indeed, suppose that, independently of the starting
position, with probability $\delta > 0$ (observe that $\delta$ might be a function of $n$ that tends to zero as $n \to \infty$) after $T(n)$ steps the distance vector decreases, where $T(n)$ is some function of $n$. By concatenating disjoint intervals of length $T(n)$, the probability
of having a strictly decreasing distance vector can be boosted as high as desired.

To show this, observe the following: if immediately after the zombies' move there is no coordinate in which all zombies have the same binary value as the survivor, then, regardless of what the survivor
does in the next round, at least one zombie will become closer to the survivor ``for free'' (that is, there exists $1 \leq j \leq k$ such that after the $t$-th round we have $d(j,t) < d(j,t-1)$).
Otherwise, suppose that there exist $1\le C \le n-1$ coordinates such that all zombies have the same binary value in this coordinate as the survivor, and the latter one can maintain the distances to all
zombies by flipping the bit corresponding to any such coordinate. In the next round, consider the following strategy: all but one zombie flip the bit recently flipped by the survivor, and the
remaining zombie flips a coordinate in which he is not the only zombie differing from the survivor in that bit. Note that this is indeed possible, since by our assumption, after the survivor's
move, at least $n/3$ zombies differ in least one bit (other than the last one flipped by the survivor), and at least $n/3$ zombies differ in least two bits (again, other than the last one flipped by
the survivor). 

Therefore, the total number of bits in which zombies differ is at least $n$, and so, by the pigeonhole principle, there exists a bit (one more time, other than the last one flipped by
the survivor) in which at least two zombies differ. With probability at least $(1/n)^k > 0$ the zombies choose this strategy, and if they do so, in the next round there are less coordinates in which
all zombies have the same binary value as the survivor. It follows that with, probability at least $(1/n)^{kC} \ge (1/n)^{k(n-1)}  > 0$, the zombies follow this sequence of strategies, and then the survivor is forced to
choose a coordinate in which the distance to at least one zombie decreases.  Since this holds independently of the distance vector, the distances eventually decrease, and the survivor is eaten with
probability 1. The proof is finished.
\end{proof}

The proof of Theorem~\ref{hyper} now follows by Lemmas~\ref{lemm1} and \ref{lemm2}. With more effort, we can obtain the order of the error term. Suppose that the survivor plays against $k$ zombies. As mentioned in the
proofs of Lemmas~\ref{lemm1} and \ref{lemm2}, $X \in \Bin(k,1/2)$ zombies start on vertices having an even number of ones in their binary representations; and $k-X$ zombies start on vertices with an odd number of ones. The random
variable $X$ determines the faith of the survivor. Since zombies at even distance to the survivor (right before they move) have more power, the survivor should choose the starting point accordingly.
If $X > k/2$, then he should choose a vertex with an even number of ones to start with; if $X < k/2$, then a vertex with an odd number of ones should be picked instead. It follows that the survivor wins if
\begin{eqnarray}
n > 2 \min\{ X, k-X \} + \max\{ X, k-X \} &=& 2 \left( \frac {k}{2} - \left| X - \frac {k}{2} \right| \right) + \left( \frac {k}{2} + \left| X - \frac {k}{2} \right| \right) \label{eq:cond_hyper} \\
&=& \frac {3k}{2} - \left| X - \frac {k}{2} \right|; \nonumber
\end{eqnarray}
and otherwise, he loses with probability 1.

Suppose that the survivor plays against $k = \frac 23 n + b \sqrt{n}$ zombies, where $b>0$ is a constant that will be determined soon. It follows from~(\ref{eq:Bin_lower}) that with probability, say, at least
0.9, $|X-k/2| \ge c\sqrt{n}$ for some small, universal, constant $c >0$. Hence, with probability at least 0.9, the condition~(\ref{eq:cond_hyper}) holds, provided that, say, $b < c/2$. On the other
hand, it follows from Chernoff's bound that with probability at least 0.9, $|X-k/2| \le b\sqrt{n}$, provided that $b$ is a large enough constant, and then the condition~(\ref{eq:cond_hyper}) fails. 

\medskip

Condition~(\ref{eq:cond_hyper}) can be also used to investigate the value of $z(Q_n)$ for small values of $n$. In particular, we find that for $n=3$ we have $z(Q_3)=2$ (in fact, $s_2(Q_3)=1/2$),
and for $n=4$, we have $z(Q_4)=3$ (in fact, $s_3(Q_4)=1/4$).

\section{Grids}

In this final section, we consider the zombie number of various grids formed by Cartesian products of graphs. We denote throughout by $G_n$ the $n \times n$ square grid, which is the graph isomorphic to
$P_n \square P_n$, where $P_n$ is the path with $n$ vertices. Our first result of the section focusses on these Cartesian grids.
\begin{theorem}
For $n\ge 2,$ we have that $z(G_n) =2$. Hence, $Z(G_n) =1$.
\end{theorem}
\begin{proof}
Since $c(G_n)=2$, it suffices to show that two zombies win the game with probability 1.  As in the proof for hypercubes, our goal is to show that, starting from any distance vector and independently from the
configuration of zombies and survivor, with probability $\delta(n) > 0$ the distance vector strictly decreases after some number of steps (which is a function of $n$). Consider the following strategy: if a zombie is not forced to move
in one direction as he shares one coordinate with the survivor (this can happen deterministically during at most $n$ consecutive steps), then he does the following: if before the survivor's last
move the zombie and the survivor shared, say, the $x$-coordinate and the survivor moved vertically, then the zombie moves then horizontally. If before the survivor's last move he did not share
neither $x$ nor $y$-coordinates, and the survivor moved horizontally (vertically, respectively), then the zombie moves also horizontally (vertically, respectively). Note that with probability at
least $(1/2)^{2n}
> 0$ both zombies follow this strategy during $n$ consecutive rounds and, as in the section devoted to hypercubes, this probability can be boosted as high as desired. If both zombies follow this strategy, we
see that after $n$ steps the distance from the survivor to the zombies strictly decreases (the survivor either gets eaten in a corner or has to move towards the zombies, decreasing distances).
Iterating the same argument, we see that the survivor gets eaten with probability $1$.
\end{proof}

An analogous strategy for two zombies can be adapted to win on a Cartesian product of trees, giving again the cost of being undead equal to 1. Indeed, since each factor of the product is a tree, there is for each
coordinate exactly one shortest path between any two vertices, and as there is no diagonal shortcut possible (that is, the shortest path from $(a_i,b_i)$ to $(a_j,b_j)$ always has to pass through
$(a_i,b_j)$ or $(a_j,b_i)$), the survivor cannot forever maintain distances in both coordinates.

\medskip

We next consider grids formed by products of cycles.  Let $T_n$ be the \emph{toroidal grid} $n\times n$, which is isomorphic to $C_n \square C_n.$ For simplicity, we take the vertex set of $T_n$ to
consist of $\Z_n \times\Z_n$. The analysis of toroidal grids is more delicate than in the Cartesian case, and we present here a lower bound for the zombie number of $T_n.$

\begin{theorem}\label{thm:torus}
Let $\omega = \omega(n)$ be a function tending to infinity as $n\to \infty$. Then a.a.s.\ $z(T_n) \ge \sqrt n/(\omega\log n)$.
\end{theorem}
In order to prove this lower bound on $z(T_n)$, we assume henceforth that there are $k=  \lfloor \sqrt n/(\omega \log n) \rfloor $ zombies, for any given $\omega = \omega(n)$ that tends to infinity as $n\to \infty$. We will find a strategy for the survivor
that allows the survivor to avoid being eaten forever a.a.s.

We introduce some formalism that will be convenient for our descriptions.
It is convenient for the analysis to assume that the game runs
forever, even if some zombie catches the survivor (in which case they will remain together forever).
A \emph{trajectory} is a sequence $\bu=(u_t)_{t\in I}$ of vertices of $T_n$, where $I$ is an interval (finite or
infinite) of non-negative integers corresponding to time-steps. We say that the survivor (or one zombie) follows a trajectory $\bu=(u_t)_{t\in I}$ if, for each $t\in I$, $u_t$ denotes the position
of that survivor or zombie at time-step $t$. Recall that zombies move first, so a zombie with zombie trajectory $\bv$ catches the survivor with trajectory $\bu$ at time-step $t$ if $v_t=u_{t-1}$.
(If $v_t=u_t$ because the survivor moves to the zombie's location, then $v_{t+1}=u_{t}$, so we may interpret this as if the zombie catches the survivor at time-step $t+1$.)
Sometimes it is useful to imagine that the survivor and the zombies move
simultaneously, but the zombies observe the position of the survivor at time $t$ to decide their new position at time $t+1$, whereas the survivor looks at the positions of the zombies at time $t+1$
do decide their new position at time $t+1$. Since the zombies' trajectories may depend on the survivor's trajectory and viceversa, it is convenient to formulate the strategy of each individual
(zombie or survivor) a priori in a way that does not depend on the other player's choices.

A \emph{zombie strategy} is given by $(v_0,\bsigma)$, where $v_0 \in\Z_n\times\Z_n$, $\bsigma =(\sigma_t)_{t\in\N}$ and each $\sigma_t$ is a permutation of the symbols $\tU,\tD,\tL,\tR$ (up, down,
left, right). Each zombie will choose a zombie strategy $(v_0,\bsigma)$ uniformly at random and independently from everything else, and this will determine the zombie's decisions throughout the
process in the following manner. Initially, the zombie starts at position $v_0$. At each step $t\in\N$, the zombie moves from $v_{t-1}$ to $v_t$ (before the survivor moves). To do so, the zombie
picks the first direction in the permutation $\sigma_t$ that decreases its distance to the survivor, and takes a step in that direction. This determines the new position $v_t$.

A \emph{survivor strategy} is given by $(u_0,\bm)$, where $u_0: (\Z_n\times\Z_n)^k \to \Z_n\times\Z_n$, $\bm = (m_t)_{t\in\N}$ and $m_t: (\Z_n\times\Z_n)^{k+1} \to \{\tU,\tD,\tL,\tR\}$.  This strategy
is chosen deterministically by the survivor before the zombie strategies have been exposed, and will determine the decisions of the survivor during the game as follows. Initially, the survivor starts
at vertex $u_0$, which is a function of the zombies' initial configuration. At each time-step $t\in\N$, after the zombies move, the survivor moves from $u_{t-1}$ to $u_t$. The direction of this move
is determined by $m_t$, which is a function of the positions of the zombies and the survivor right before their move.
Note that the strategy of the survivor depends not only on the positions of all players at a given time-step $t$, but also on $t$.
The dependency on $t$ does not provide an essential advantage for the survivor, but makes the description of the argument easier.

The formulation above is useful, since it allows us to decouple all the decisions by the survivor and the zombies prior to the start of the game. Note that the final trajectory of each individual
(survivor or zombie) will depend not only on their own strategy, but will be a deterministic function of all strategies together.

Throughout the section, let $B$ be a fixed $\lfloor K\log n\rfloor \times \lfloor K\log n\rfloor$ box contained in the toroidal grid, where $K= 5 \cdot 10^4$.
A survivor strategy is \emph{$B$-boxed} during the
time period $[0,4n]$ if the following hold: the initial position $u_0$ belongs to the box $B$ and is chosen independently of the positions of the zombies; the sequence of moves
$\bm=(m_t)_{t\in[1,4n]}$ is such that the survivor always stays inside of $B$, regardless of the positions of the zombies in that period; each move $m_t$ ($t\in[1,4n]$) does not depend on the
positions of the zombies that lie outside of $B$ at that given step $t$ (that is, any two configurations of the zombies at time $t$ that only differ in the positions of some zombies not in $B$ must
yield the same value of $m_t$). Later in this section, we will specify a particular $B$-boxed strategy for the survivor that will allow him to survive forever a.a.s.\ The next two lemmas describe the typical behaviour of the zombies before they reach the box $B$, by only assuming that the survivor's strategy is $B$-boxed during the time period $[0,4n]$.

\begin{lemma}\label{lem:arrivalB}
Assume that the survivor's strategy is $B$-boxed during the time period $[0,4n]$, and pick any zombie strategy for all but one distinguished zombie. For any $t\in[1,4n]$, the probability that this
zombie is initially outside of the box $B$ and arrives at $B$ at the $t$-th step of the game is at most $20Kt\log n/n^2$.
\end{lemma}
\begin{proof}
Fix any $B$-boxed strategy for the survivor and the zombie strategies for the remaining zombies. Expose the sequence $\bsigma$ of move priorities of our distinguished zombie, but not his initial
position $v_0$. Our goal is to show that there is a set of vertices $V_0$ of order $|V_0|\le 20Kn\log n$ such that the event that the zombie arrives at $B$ at step $t$ (and not before) implies that
$v_0\in V_0$. Note that, by the definition of a $B$-boxed strategy, in the event that our zombie reaches $B$ at step $t$ for the first time, then the trajectory $\bu = (u_t)_{t\in[0,t-1]}$ of the
survivor during the time period before step $t$ does not depend on the behaviour of that one zombie, and we can regard $\bu$ as a fixed sequence.

Given any vertex $v\in (\Z_n\times\Z_n) \setminus B$ and $i\in\{1,2,\ldots, t\}$, define $\nu_i(v)$ to be the new position of the zombie at step $i$ assuming it was on vertex $v$ at the end of step
$i-1$ and that the survivor was at $u_{i-1}$. Note that $\nu_i$ is well defined on $(\Z_n\times\Z_n) \setminus B$, given our choice of $\bu$ and $\bsigma$. Let $w$ be any fixed vertex in the inner
boundary of $B$ (that is, $w\in B$ and $w$ is adjacent to some vertices not in $B$). Suppose that the zombie arrives at vertex $w$ at step $t$ and is outside of $B$ before that. For
$i\in\{0,1,\ldots,t-1\}$,  let $V_i(w)$ be the set of vertices in $(\Z_n\times\Z_n) \setminus B$ to which the zombie may move at step $i$. Also, $V_t(w)=\{w\}$ by our assumption, and
${\nu_i}^{-1}  ( V_i(w) ) = V_{i-1}(w)$ for all $1\le i\le t$. We will show that $|V_{0}(w)| \le 4t+1$.

We say that two vertices are \emph{horizontally} (or \emph{vertically}) \emph{aligned} if they share the same  horizontal (vertical) coordinate. Moreover, we say they are aligned if they are either
horizontally or vertically aligned. For $1\le i\le t$, it is immediate to check that if $v$ and $u_{i-1}$ are not aligned, then $|{\nu_i}^{-1}(v)|\le1$. Observe that ${\nu_i}^{-1}(v)$ could be empty
if, for instance, the horizontal distance between $v$ and $u_{i-1}$ is $\lfloor n/2\rfloor$ and $\sigma_i=(\tL,\tR,\tU,\tD)$. Otherwise, if $v$ and $u_{i-1}$ are aligned, then $|{\nu_i}^{-1}(v)|\le
3$ (for instance, if $v$ is horizontally aligned with $u_{i-1}$ and to the left of $B$, then ${\nu_i}^{-1}(v)$ must be either one step above, below or to the left of $v$). We will prove that for
each $1\le i\le t$, $V_i(w)$ contains at most one vertex that is horizontally aligned with $u_{i-1}$ and, similarly, at most one vertex that is vertically aligned with $u_{i-1}$. This claim is
obvious for $i=t$, since $|V_t(w)|=1$. 

For $1\le i\le t-1$, suppose that there are two vertices $v,v'\in V_i(w)$ that are horizontally aligned with $u_{i-1}$. Then they both must be on the same side
(left or right) of $B$, depending on the position of $w$. This follows since a zombie cannot escape from the horizontal strip of dimensions $\lfloor K\log n\rfloor \times \lfloor (n-\lfloor K\log n\rfloor)/2 \rfloor$ to the left of $B$ and the same for the symmetric strip to the right of $B$; that is, ~$\nu_i$ maps these strips into themselves. Without loss of generality, we may assume that both $v$ and $v'$ are to the left of $B$. Also, $v$ and $v'$ must belong to the horizontal strip $S$ of
the same height of $B$ containing $B$. Suppose that $v$ and $v'$ are at distance $d$ from each other. Then $\nu_{i+1}(v)$ and $\nu_{i+1}(v')$ must be horizontally aligned, contained in $S$, to the
left of $B$ (or at $w$) and at the same distance $d$ from each other. Inductively, $\nu_t \circ \cdots \circ \nu_{i+2} \circ \nu_{i+1}(v) = w$ and $\nu_t \circ \cdots \circ \nu_{i+2} \circ
\nu_{i+1}(v') = w$ must also be at distance $d$ from each other, so $d=0$ and $v=v'$, as desired. An analogous argument shows that  $V_i(w)$ contains at most one vertex vertically aligned with
$u_{i-1}$. Therefore, for every $1\le i\le t$,
\[
|V_{i-1}(w)| = |{\nu_i}^{-1}(V_i(w))| \le |V_{i}(w)|+4,
\]
since every element in $V_i(w)$ has at most one preimage except for possibly two elements that have at most three. Hence,
$|V_{i}(w)| \le 4(t-i) +1$, and so $|V_{0}(w)| \le 4t+1$.

Let $V_0 = \bigcup_{w\in I} V_0(w),$ where the union is taken over the set $I$ of vertices of the inner boundary of $B$. Since there are at most $4K\log n$ choices for $w$ at the inner boundary of $B$, we have that $|V_0| \le (4t+1)4K\log n \le 20Kt\log n$. By construction, if the zombie
reaches the box $B$ at step $t$ for the first time, then his initial position $v_0$ must belong to $V_0$. This event happens with probability $|V_0|/n^2 \le 20Kt\log n/n^2$.
\end{proof}
\begin{lemma}\label{lem:usefultorus}
Consider $k= \lfloor \sqrt n/(\omega \log n) \rfloor$ zombies on $T_n$, for any given $\omega = \omega(n)$ that tends to infinity as $n\to \infty$. Assume that the survivor follows a $B$-boxed strategy during the time period $[0,4n]$. Then a.a.s.\ the
following hold:
\begin{enumerate}
\item[(i)]
there is no zombie in $B$ initially;
\item[(ii)]
all zombies arrive to $B$ within the first $3n$ steps; and
\item[(iii)]
no two zombies arrive at $B$ less than $M\log n$ steps apart, where $M=12K$.
\end{enumerate}
\end{lemma}
\begin{proof}
The expected number of zombies in $B$ at the initial step is
\[
k|B|/n^2 \le \frac{K^2 \sqrt n \log^2n}{\omega n^2 \log n} = o(1),
\]
so by Markov's inequality, part~(i) holds.

\bigskip

Given a zombie with (random) zombie strategy $(v_0,\bsigma)$, let $X_\tL$ be the number of steps $i\in[1,3n]$ such that $\sigma_i$ has $\tL$ as the first symbol in the permutation. Define
$X_\tR,X_\tU,X_\tD$ analogously. Observe that if $X_\tL,X_\tR,X_\tU,X_\tD>n/2$, then the zombie must reach $B$ within the first $3n$ steps, deterministically and regardless of his initial position
on $T_n$. Each of the random variables $X_\alpha$ ($\alpha\in\{\tL,\tR,\tU,\tD\}$) is distributed as $\Bin(3n,1/4)$. Therefore, by Chernoff's bound,
\[
\Prob(X_\alpha \le n/2) \le e^{- n/24}
\]
By taking a union bound over all zombies and $\alpha\in\{\tL,\tR,\tU,\tD\}$, we conclude that a.a.s.\ for every zombie $X_\tL,X_\tR,X_\tU,X_\tD>n/2$. Consequently, a.a.s.\ all zombies must reach
$B$ within the first $3n$ steps, and the proof of part~(ii) is finished.

\bigskip

In order to prove part~(iii), we consider two zombies. Suppose that the first zombie is initially not in $B$ and reaches $B$ at step $t\in[1,3n]$. By Lemma~\ref{lem:arrivalB}, the probability that
the second zombie arrives at $B$ at step $t'\in[t-M\log n,t+M\log n]$ is $O(\log^2 n/n)$. Therefore, the probability that a pair of zombies arrive at $B$ within the first $3n$ steps and less than
$M\log n$ steps apart is $O(\log^2 n/n)$. Taking a union bound over the number $\binom{k}{2} = o(n/\log^2n)$ of pairs of zombies and in view of part~(ii), we conclude the proof of part~(iii).
\end{proof}
A zombie strategy $(v_0,\bsigma)$ is called \emph{regular} if, for any direction $\alpha\in\{\tL,\tR,\tU,\tD\}$ and any interval of consecutive steps $I \subseteq [1,4n]$ of length $ \lfloor 20\log n \rfloor $, there is a subset of
steps $J\subseteq I$ (not necessarily consecutive) with $|J|= \lceil \log n \rceil$ such that, for every $i\in J$, $\sigma_i$ has $\alpha$ as the first symbol in the permutation. Informally, for every
$ \lfloor 20\log n \rfloor $ consecutive steps in $[1,4n]$, there are at least $\log n$ steps in which the zombie ``tries'' to move in the direction of $\alpha$ if that decreases the distance to the survivor.
\begin{lemma}\label{lem:regulartorus}
Consider $k=\sqrt n/(\omega \log n)$ zombies on $T_n$, for any given $\omega = \omega(n)$ that tends to infinity as $n\to \infty$. Then a.a.s.\ every zombie has a regular zombie strategy.
\end{lemma}
\begin{proof}
Given a zombie, a symbol $\alpha\in\{\tL,\tR,\tU,\tD\}$ and an interval of steps $I$ of length $\lfloor 20\log n \rfloor$, the number of steps $i$  in $I$ such that $\sigma_i$ ($i\in I$) has $\alpha$ as its first
symbol is distributed as $\Bin(\lfloor 20\log n \rfloor,1/4)$. Hence, by Chernoff's bound, the probability of having less than $\log n$ of such steps is at most
\[
e^{-((1+o(1))/2)(4/5)^25\log n} = n^{-8/5+o(1)} = o(n^{-3/2}).
\]
The proof of the lemma follows by taking a union bound over all $k=o(n^{1/2})$ zombies, all $O(n)$ choices of $I$ and all $\alpha\in\{\tL,\tR,\tU,\tD\}$.
\end{proof}
Given a trajectory $\bu=(u_t)_{t\in[a,b]}$, any integer $a<j<b$ such that $u_j-u_{j-1} \ne u_{j+1}-u_j$ is called a \emph{turning point} (that is, the direction of the trajectory changes at time-step
$j+1$). A turning point $j$ is \emph{proper} if additionally $u_j-u_{j-1} \ne - (u_{j+1}-u_j)$ (that is, the trajectory does not turn $180^\circ$). For convenience, we also consider the first and last
indices $a$ and $b$ of the trajectory to be proper turning points. We call a trajectory to be \emph{stable} if all its turning points are proper and, for any two different turning points $j$ and
$j'$, we have $j-j'\ge \lfloor 20\log n \rfloor $.

\begin{lemma}\label{lem:followtorus}
Suppose that a zombie has a fixed regular zombie strategy $(v_0,\bsigma)$ and that the survivor follows a stable trajectory $\bu$ during the time interval $[a,b]$. Let $\bv$ be the trajectory of the
zombie (determined by its strategy and the survivor's trajectory). Suppose moreover that $v_a$ and $u_a$ are at distance $d\in\{2,3\}$ and that $v_{a+1}$ and $u_{a+1}$ are also at distance $d$ (that
is, the first move of the survivor is not towards the zombie). Then, deterministically, $v_t$ and $u_t$ are at distance $d$ for all $t\in[a,b]$.
\end{lemma}
\begin{proof}
Let $a'>a$ be the first turning point in $[a,b]$ after $a$. It is clear that, since the survivor is not changing direction between $u_a$ and $u_{a'}$ and he is not going towards the zombie, the
distance between $u_t$ and $v_t$ stays constant for all $t\in[a,a']$. If $a'=b$, then we are done. Otherwise, if $a'<b$, then we need to guarantee that the survivor does not move towards the zombie at time-step $a'+1$.

Without loss of generality, we may assume that $u_a=(0,0)$ and $u_{a+1}=(1,0)$ (that is, the survivor moves to the right between time-steps $a+1$ and $a'$). If the zombie is initially horizontally aligned with the survivor (that
is, ~$v_0=(-2,0)$ or $v_0=(-3,0)$), then it must stay so between time-steps $a+1$ and $a'$. Then the survivor will move away from the zombie at time-step $a'+1$ as well, regardless of the survivor's
choice of new direction, since the turning point $a'$ is proper. Otherwise, if the zombie is above the survivor at time-step $a$ (that is, $v_0 \in \{(-2,1), (-1,2), (0,3), (-1,1), (0,2)\}$),
then it must become horizontally aligned with them before they reach $u_{a'}$ since, by stability of the survivor's strategy, $a' - a \geq \lfloor 20\log n \rfloor$, and also due to the regularity of the zombie strategy, the zombie will have $\tD$ as a priority direction move for at least
$\log n>3$ time-steps in that time period. The case in which the zombie is initially below the survivor is treated analogously. 

Summarising the arguments, in either case the zombie is horizontally
aligned with the survivor before time-step $a'$, and thus, the survivor can safely change direction at time-step $a'+1$ without moving towards the zombie. Finally, we can inductively repeat the argument and prove the statement for the whole time period $[a,b]$.
\end{proof}
\begin{proof}[Proof of Theorem~\ref{thm:torus}]
We will describe a strategy for the survivor during the time period $[0,4n]$ that a.a.s.\ succeeds at attracting all the zombies to two vertices at distances $2$ and $3$ from the survivor's position
and on the same side of that position. Once that is achieved, the survivor can simply keep moving in a straight line around the toroidal grid, staying away from all the zombies forever; see Figure~\ref{fig:forever}.
\begin{figure}
\includegraphics{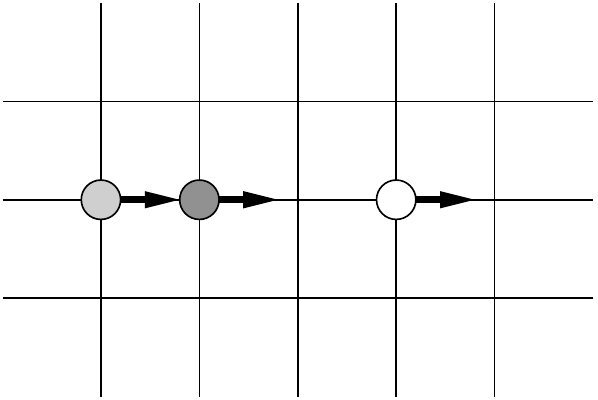}
\caption{If at some point the survivor is on the white vertex and all the zombies are at distance $2$ or $3$ on the grey vertices, and the survivor is moving away from the zombies,
then he can keep the same direction forever on $T_n$ and survive.}
\label{fig:forever}
\end{figure}

Let $C$ be a smaller $ \lfloor 20\log n \rfloor  \times \lfloor 20\log n \rfloor $ box centered at the center of box $B$. The survivor starts at the top left corner $u_0$ of $C$, and will always follow a stable trajectory $\bu$
during the time period $[0,4n]$ inside $B$. The survivor's decisions regarding what trajectory to follow, will depend on the positions of the zombies inside $B$, but not on those outside of $B$. Therefore, the
survivor strategy is $B$-boxed during the time period $[0,4n]$.

In our description of the survivor's strategy, we will only consider situations that are achievable assuming that the conclusions of Lemma~\ref{lem:usefultorus} and Lemma~\ref{lem:regulartorus} hold. That is, we assume that initially there is no zombie in $B$; they all arrive at $B$ within the first $3n$ steps; no two zombies arrive at $B$ less than $M\log n$ steps apart; and all zombies have
regular strategies. If at some step the survivor has to face a situation not covered by our description (because, for instance, two zombies arrived at $B$ at the same time), then he gives up and simply defaults to
any arbitrary fixed $B$-boxed strategy, ignoring the zombies' behaviour from then on. As a result the survivor will probably be eaten but, fortunately, this situation does not happen a.a.s.

The survivor starts at the top left corner $u_0$ and starts going in circles clockwise around $C$ until the time a first zombie arrives at $B$. Let $a$ be the time this situation occurs. The survivor keeps
going in circles around $C$ until the zombie is at distance  less or equal to $542\log n$. From our assumption on the regularity of zombies' strategies, this takes at most $  (K/2) \lfloor 20 \log n \rfloor \le  10K\log n$ steps from time $a$ (since
there will be at least $(K/2) \log n$ steps among those in which the zombie tries to move in the direction of $\alpha$, for each $\alpha\in\{\tL,\tR,\tU,\tD\}$). Then the survivor keeps going until
the next corner $u$ of $C$. At that point the zombie is at distance between $542\log n$ and $500\log n$ from him. The survivor makes that corner $u$ his next proper turning point in his
trajectory, and changes the direction in a way that he is not moving towards the zombie. The survivor can always do so by choosing between a $90^\circ$ or a $-90^\circ$ turn.  (Notice that he might leave $C$ at $u$.)

Without loss of generality, we may suppose that
this direction is to the right (the description of his strategy in any other case is analogous). The survivor keeps moving right for $ 1000 \lfloor 20 \log n \rfloor \le  2 \cdot 10^4 \log n$ time-steps. During those steps, the zombie gets
horizontally aligned with the survivor (and it is still at the same distance), since its zombie strategy is regular and at least $1000\log n$ of those steps decrease the vertical distance between the
two individuals. Then the survivor goes down for $ \lfloor 20\log n \rfloor $ steps, left for $ \lfloor 20\log n \rfloor $ steps and up again for $ \lfloor 20\log n \rfloor $ steps. The zombie is still to the left of the survivor (at horizontal
distance of between $440\log n$ and $542\log n$) and either horizontally aligned or below (at vertical distance between $0$ and $ \lfloor 20\log n \rfloor $). Next, the survivor moves to the left $ 20\lfloor 20\log n \rfloor \le 400\log n$ steps.
Both individuals must now be horizontally aligned and at distance between $40\log n$ and $143\log n$. The survivor keeps on moving left until he is at distance $2$ or $3$ from the zombie. Let $b$
be the time-step when this happens. Finally, the survivor moves down for $ \lfloor 20\log n \rfloor$ steps, and then moves back the top left corner $u_0$ of $C$ using a stable trajectory. He chooses the shortest
stable trajectory that allows the survivor to reach $u_0$ by a step up. Let $c$ be the time he gets back at $u_0$. Then he resumes the strategy of going around $C$ clockwise until the next zombie arrives to $B$.
See Figure~\ref{fig:stablestrategy} for a visual representation of the above description.

\begin{figure}
\includegraphics{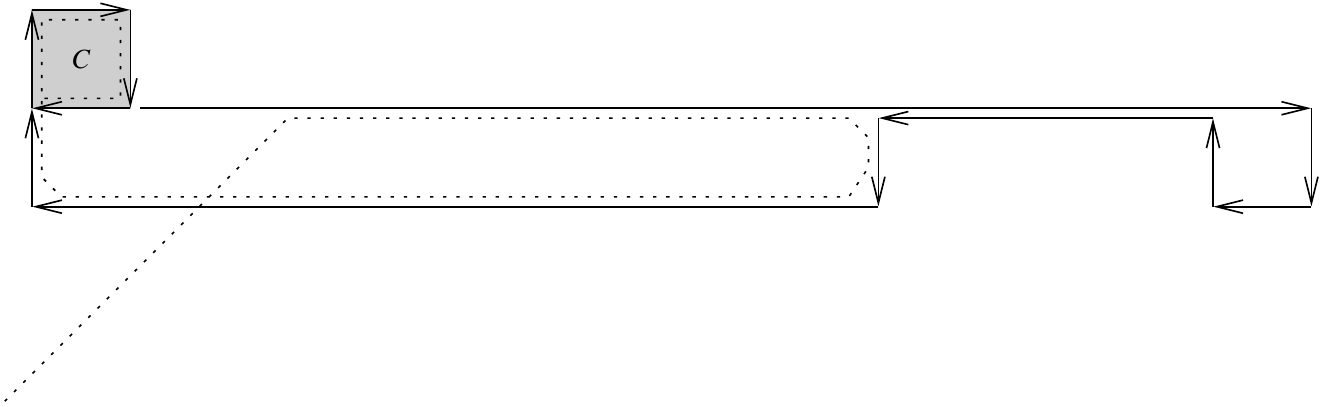}
\caption{Approximate depiction of the survivor's strategy when a new zombie approaches. The black arrows describe the trajectory of the survivor and the dotted lines the trajectory of the zombie.}
\label{fig:stablestrategy}
\end{figure}

Note that the survivor's trajectory described so far is stable, since all the turning points are proper and are at least $\lfloor 20\log n \rfloor$ steps apart. Further, the survivor move down at step $b+1$ is
not towards the zombie (which is at time $b$ horizontally aligned with the survivor and to his left). Therefore, by Lemma~\ref{lem:followtorus}, as long as the survivor maintains a stable
trajectory over the whole time period $[0,4n]$, then the trajectory of this zombie will keep a constant distance (either $2$ or $3$) to the survivor's trajectory during all steps in $[b,4n]$. Also, observe
that the whole process between time $a$ and time $c$ takes at most $11K\log n < M\log n$ steps, so by assumption there was no other zombie in $B$ during that period. Moreover, the survivor's strategy
does not depend on zombies outside of $B$ and, in spite of his long excursion of around $2 \cdot 10^4 \log n$ steps to the right from one corner of $C$, the survivor never abandons the box $B$, as required by the
definition of $B$-boxed strategy.

The survivor can proceed analogously each time a new zombie arrives to $B$, ignoring all zombies that are already at distance $2$ or $3$ from them. By construction, this defines a $B$-boxed strategy during the time period $[0,4n]$ (recall that for any configuration not covered by our previous description, the survivor just adopts any default $B$-strategy). Moreover, if all our
assumptions on the zombies hold (which occurs a.a.s.~by Lemmas~\ref{lem:usefultorus} and~\ref{lem:regulartorus}), the survivor will follow a stable trajectory during the time period $[0,4n]$ with
the following properties: at step $4n$ all zombies are at distance $2$ or $3$ from the survivor, and the survivor is moving away from all of them. Then from that moment on, the survivor can keep
going in the same direction and survive forever (deterministically). The proof of the theorem is finished.
\end{proof}

A sub-quadratic upper bound for the zombie number of toroidal grids remains open. We plan to consider this problem, and the zombie number of other grid graphs, in future work.

\end{document}